\documentclass[a4paper,10pt]{article}

\usepackage{graphicx}
\usepackage{geometry}
\usepackage{amsthm}
\usepackage{amssymb}
\usepackage{amsmath}
\usepackage{hyperref}
\usepackage{xcolor}
\usepackage{enumerate}
\usepackage{listings}
\usepackage[textsize=tiny]{todonotes}
\usepackage{complexity}
\usepackage{nicefrac}
\usepackage{wrapfig}
\usepackage{lineno}
\usepackage{verbatim}

\usepackage[font=small,labelfont=bf]{caption}
\usepackage[font=small,labelfont=normalfont,labelformat=simple]{subcaption}

\graphicspath{{figs/}}

\newtheorem{theorem}{Theorem}

\newtheorem{lemma}[theorem]{Lemma}

\newtheorem{question}{Question}

\DeclareMathOperator{\area}{area}

\DeclareMathOperator{\dist}{dist}

\def\EE{\mathbb{E}}

\def\RR{\mathbb{R}}
\def\SS{\mathbb{S}}

\def\AA{{\cal A}}
\def\BB{{\cal B}}
\def\CC{{\cal C}}
\def\FF{{\cal F}}
\def\HH{{\cal H}}
\def\II{{\cal I}}
\def\LL{{\cal L}}

\def\Zj{V_{\le j}}

\newcommand{\risingfactorial}[1]{^{\overline{#1}}}

\def\inst#1{$^{#1}$}

\date{}

\title{On the Average Complexity of the $k$-Level\thanks{
    M.-K.\ Chiu was supported  by ERC StG 757609. S.\ Felsner
    and M.\  Scheucher were supported  by DFG Grant  FE~340/12-1.  R.\
    Steiner was supported by DFG-GRK 2434. 
    P.\ Schnider was supported by the SNSF Project 200021E-171681.
    P.\ Valtr was supported by
    the grant no.~18-19158S of the Czech Science Foundation (GA\v{C}R)
    and by  the PRIMUS/17/SCI/3  project of Charles  University.  This
    work was initiated at a workshop of the collaborative DACH project
    \emph{Arrangements and  Drawings} in Schloss St.~Martin,  Graz. We
    thank the organizers for the  inspiring atmosphere.  We also thank
    Birgit Vogtenhuber for interesting discussions.}}

\begin{document}

\author{
Man-Kwun Chiu\inst{1}
\and
Stefan Felsner\inst{2}
\and 
Manfred Scheucher\inst{2}
\and
Patrick Schnider\inst{3}
\and 
Raphael Steiner\inst{2}
\and
Pavel Valtr\inst{4}
}

\maketitle

\begin{center}
{\footnotesize
\inst{1}
Department of Mathematics and Computer Science, \\
Freie Universit\"at Berlin, Germany, \\
\texttt{\{chiumk\}@zedat.fu-berlin.de}
\\\ \\
\inst{2} 
Institut f\"ur Mathematik, \\
Technische Universit\"at Berlin, Germany,\\
\texttt{\{felsner,scheucher,steiner\}@math.tu-berlin.de}
\\\ \\
\inst{3} 
Department of Computer Science, \\
ETH Z\"urich, Switzerland\\
\texttt{\{patrick.schnider\}@inf.ethz.ch}
\\\ \\
\inst{4} 
Department of Applied Mathematics, \\
Faculty of Mathematics and Physics, Charles University, Czech Republic \\
\texttt{\{valtr\}@kam.mff.cuni.cz}
}
\end{center}

\begin{abstract}

\noindent
Let $\LL$ be an arrangement of $n$ lines in the Euclidean plane.  The
\emph{$k$-level} of $\LL$ consists of all vertices $v$ of the
arrangement which have exactly $k$ lines of $\LL$ passing below $v$.
The complexity (the maximum size) of the $k$-level in a line
arrangement has been widely studied.  In 1998 Dey proved an upper
bound of $O(n\cdot (k+1)^{1/3})$.  Due to the correspondence between
lines in the plane and great-circles on the sphere, the asymptotic
bounds carry over to arrangements of great-circles on the sphere,
where the $k$-level denotes the vertices at distance at most $k$ to a
marked cell, the \emph{south pole}.

We prove an upper bound of $O((k+1)^2)$ on the expected
complexity of the $k$-level in great-circle arrangements 
if the south pole is chosen uniformly at random
among all cells. 

We also consider arrangements of great $(d-1)$-spheres on the sphere
$\SS^d$ which are orthogonal to a set of random points on $\SS^d$. In
this model, we prove that the expected complexity of the $k$-level is
of order $\Theta((k+1)^{d-1})$.
\end{abstract}

\section{Introduction}
\label{sec:intro}

Let $\LL$ be an arrangement of $n$ lines in the Euclidean plane.  The
\emph{vertices} of $\LL$ are the intersection points of lines of $\LL$.
Throughout this article we consider arrangements with the properties that
no line is vertical and 
no three lines intersect in a common vertex.
The \emph{$k$-level} of $\LL$ consists of
all vertices $v$ which have exactly $k$ lines of $\LL$ below $v$.  We denote
the $k$-level by $V_k(\LL)$ and its size by $f_k(\LL)$.  Moreover, by $f_k(n)$
we denote the maximum of $f_k(\LL)$ over all arrangements $\LL$ of $n$ lines,
and by $f(n)=f_{\lfloor (n-2)/2 \rfloor}(n)$ the maximum size of the
\emph{middle level}.

A \emph{$k$-set} of a finite point set $P$ in the Euclidean plane is a subset
$K$ of $k$ elements of $P$ that can be separated from $P\setminus K$ by a
line. Paraboloid duality is a bijection $P \leftrightarrow \LL_P$ between
point sets and line arrangements (for details on this duality
see~\cite[Chapter~6.5]{ORourke1994_book} or
\cite[Chapter~1.4]{Edelsbrunner1987_book}).  The number of $k$-sets of $P$
equals $|V_{k-1}(\LL_P) \cup V_{n-1-k}(\LL_P)|$.

In discrete and computational geometry bounds on the number of $k$-sets of a
planar point set, or equivalently on the size of $k$-levels of a planar line
arrangement have important applications. 
The complexity of $k$-levels was first studied by Lov\'{a}sz~\cite{Lovasz1971}
and Erd\H{o}s et al.~\cite{ErdosLSS1973}. 
They bound the size of the $k$-level
by $O(n \cdot (k+1)^{1/2})$.  Dey \cite{Dey1998} used the crossing lemma to
improve the bound to $O(n \cdot (k+1)^{1/3})$.  In particular, the maximum size
$f(n)$ of the middle level is $O(n^{4/3})$.  Concerning the lower bound on the
complexity, Erd\H{o}s et al.~\cite{ErdosLSS1973} gave a construction showing
that $f(2n) \ge 2f(n) + cn = \Omega(n \log n)$ and conjectured that
$f(n) \ge \Omega(n^{1+\varepsilon})$.  An alternative 
$\Omega(n \log n)$-construction was given by Edelsbrunner and
Welzl~\cite{EdelsbrunnerWelzl1985}.  The current best lower bound
$f_k(n) \ge n \cdot e^{\Omega(\sqrt{\log k})}$ was obtained by
Nivasch~\cite{Nivasch2008} improving on a bound by T\'oth~\cite{Toth2001}. 


\subsection{Generalized Zone Theorem}

In order to define ``zones'', 
let us introduce the notion of ``distances''.
For
$x$ and $x'$ being a vertex, edge, line, or cell of an arrangement $\LL$ of lines in $\RR^2$
we let their \emph{distance} $\dist_{\LL}(x,x')$ be the minimum number of lines of
$\LL$ intersected by the interior of a curve connecting a point of $x$ with a
point of $x'$.  
Pause to note that the $k$-level of $\LL$ is precisely
the set of vertices which are at distance $k$ to the bottom cell.

The $(\leq j)$-zone
$Z_{\leq j}(\ell, \LL)$ of a line $\ell$ in an arrangement $\LL$ is defined as
the set of vertices, edges, and cells from $\LL$ which have distance at most
$j$ from $\ell$. See Figure~\ref{Fig:Zones} for an illustration.

\begin{figure}[htb]
  \centering
  
  \hbox{}\hfill
    \begin{subfigure}[t]{.58\textwidth}
    \centering
    \includegraphics[width=0.95\textwidth]{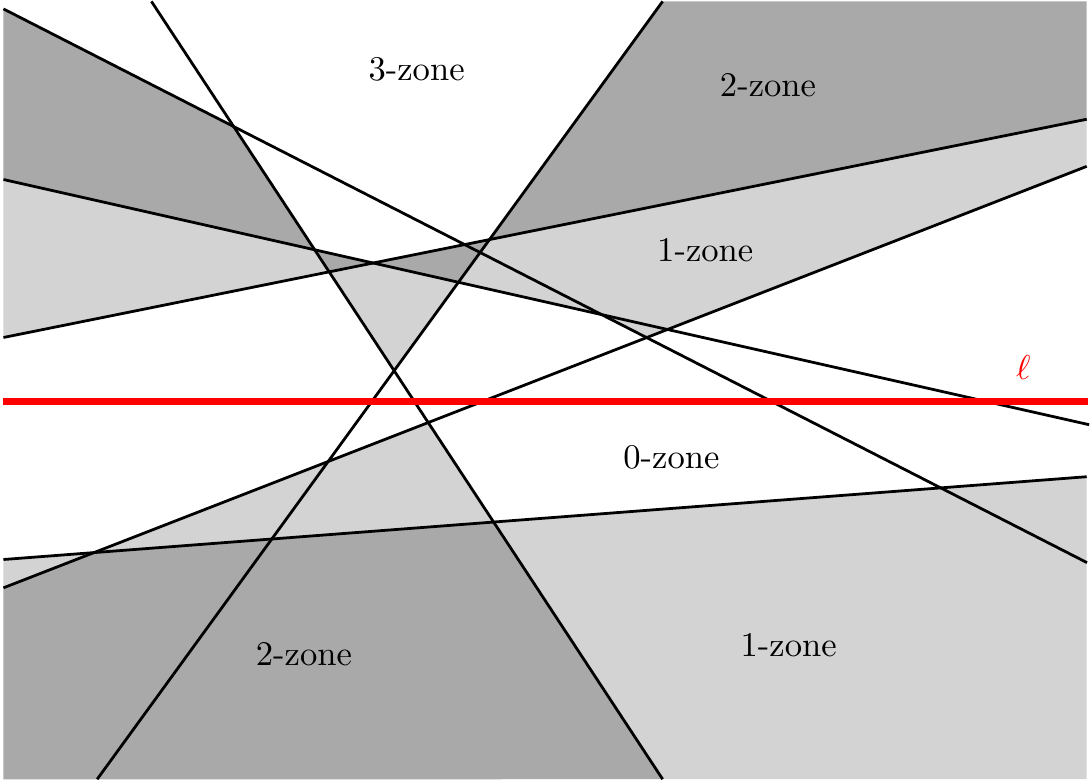}
    \caption{}
    \label{Fig:Zones}
  \end{subfigure}
  \hfill
  \begin{subfigure}[t]{.24\textwidth}
    \centering
    \includegraphics[width=0.95\textwidth]{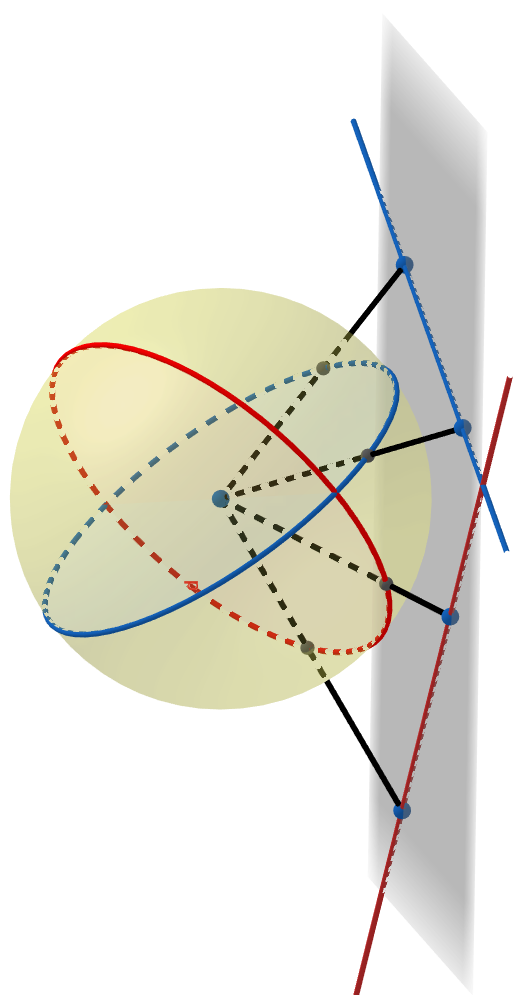}
    \caption{}
    \label{fig:duality}
  \end{subfigure}
  \hfill\hbox{}

  \caption{
  \subref{Fig:Zones} The higher order zones of a line $\ell$.
  \subref{fig:duality} The correspondence between great-circles 
    on the unit sphere and lines in a plane. 
    Using the center
    of the sphere as the center of projection points on the sphere 
    are projected to the points in the plane.
  }
\end{figure}

For arrangements of hyperplanes in $\RR^d$ the $(\leq j)$-zone is
defined alike. The classical zone theorem provides bounds for the zone
($(\leq 0)$-zone) of a hyperplane (cf.\
\cite{EdelsbrunnerSeidelSharir1991} and
\cite[Chapter~6.4]{Matousek2002_book}). A generalization with bounds
for the complexity of the $(\leq j)$-zone
appears as an exercise in Matou\v{s}ek's book
\cite[Exercise~6.4.2]{Matousek2002_book}. In the proof of
Theorem~\ref{Thm:average_k-level} we use a variant of the
2-dimensional case (Lemma~\ref{thm:ckn_bound_higher}). For the sake
of completeness and to provide explicit constants, we include the
proof 
in Section~\ref{sec:proof_ckn_bound}.

\begin{lemma}\label{thm:ckn_bound_higher}
  Let $\LL$ be a simple arrangement of $n$ lines in $\RR^2$ and
  $\ell \in \LL$. The $(\le j)$-zone of $\ell$ contains at most
  $2e \cdot (j+2)n$ vertices strictly above $\ell$.
\end{lemma}

\subsection{Arrangements of Great Circles}
\label{ssec:arrangements_greatcircles}

Let $\Pi$ be a plane in 3-space which does not contain the origin and
let $\SS^2$ be a sphere in 3-space centered at the origin. The
central projection $\Psi_\Pi$ yields a bijection  between arrangements
of great circles on $\SS^2$ and arrangements of
lines in $\Pi$.
Figure~\ref{fig:duality} gives an illustration.

The correspondence $\Psi_\Pi$
preserves intersesting properties, e.g. simplicity of the
arrangements. If $\Psi_\Pi(\CC) = \LL$, and $\LL$ has no parallel lines, then
$\Psi_\Pi$ induces a bijection between pairs of antipodal vertices
of~$\CC$ and vertices of~$\LL$. 

As in the planar case, 
we define the \emph{distance} between points $x,y$ of $\SS^2$ relative to a
great-circle arrangement $\CC$ as the minimum number of circles of
$\CC$ intersected by the interior of a curve connecting $x$ with $y$.
The \emph{$k$-level} (\emph{$\le k$-zone} resp.) of~$\CC$ is 
the set of all the vertices of~$\CC$ 
at distance $k$ (distance at most $k$ resp.) from the south pole.

Let $\Pi_1$ and $\Pi_2$ be two parallel planes in 3 space with the
origin between them and let $\Psi_1$ and $\Psi_2$ be the respective
central projections. For a great-circle arrangement $\CC$ 
we consider $\LL_1=\Psi_1(\CC)$ and $\LL_2=\Psi_2(\CC)$.
A vertex $v$ from the $k$-level of $\CC$
maps to a vertex of the $k$-level in one of $\LL_1$, $\LL_2$ and to a
vertex of the $(n-k-2)$-level in the other. Hence,
bounds for the maximum size of the $k$-level of line arrangements 
carry over to the $k$-level of great-circle arrangements 
except for a multiplicative factor of~2.

The $(\le j)$-zone of a great-circle $C$ in $\CC$ projects
to a $(\le j)$-zone of a line in each of $\LL_1$ and $\LL_2$.
Hence, the complexity of a $(\le j)$-zone in $\CC$ is upper
bounded by two times the maximum complexity of a  $(\le j)$-zone in a
line arrangement. Lemma~\ref{thm:ckn_bound_higher}
implies that the $(\le j)$-zone of a great-circle $C$ 
in an arrangement of $n$ great-circles 
contains at most $4e \cdot (j+2) n$ vertices.

\subsection{Higher Dimensions}

The problem of determining the complexity of the $k$-level admits a natural
extension to higher dimensions. We consider arrangements in $\RR^d$
of hyperplanes with the properties that 
no hyperplane is parallel to the $x_d$-axis and
no $d+1$ hyperplanes intersect in a common point. 
The \emph{$k$-level} $V_k(\AA)$ of $\AA$ consists of all
vertices (i.e. intersection points of $d$ hyperplanes) which have
exactly $k$ hyperplanes of $\AA$ below them (with respect to the $d$-th
coordinate).  We denote the $k$-level by $V_k(\AA)$ and its size by
$f_k(\AA)$.  Moreover, by $f_k^{(d)}(n)$ we denote the maximum of $f_k(\AA)$
among all arrangements $\AA$ of $n$ hyperplanes in $\RR^d$.

As in the planar case, there remains a gap between lower and upper bounds;
\[
\Omega(n^{\lfloor d/2 \rfloor}k^{\lceil d/2 \rceil -1})
\leq 
f_k^{(d)}(n)
\leq 
O(n^{\lfloor d/2 \rfloor}k^{\lceil d/2 \rceil - c_d}),
\] 
here $c_d>0$ is a small positive constant only depending on~$d$.  Details and
references can be found in Chapter~11 of Matou\v{s}ek's
book~\cite{Matousek2002_book}.
In dimensions $3$ and~$4$ improved bounds have
been established. For example, for $d=3$, it is known that
$ f_k^{(3)}(n) \le O(n(k+1)^{3/2}) $ (see \cite{SharirSmorodinskyTardos2001}).
For the middle level in dimension $d \ge 2$ an improved lower bound
$ f^{(d)}(n) \ge n^{d-1}\cdot e^{\Omega(\sqrt{\log n})} $ is known (see
\cite{Toth2001} and \cite{Nivasch2008}).

We call the intersection of $\SS^d$ with a central hyperplane
in $\RR^{d+1}$ a \emph{great-$(d-1)$-sphere} of $\SS^d$. 
Similar to the planar case, 
arrangements of hyperplanes in $\RR^d$
are in correspondence with 
arrangements of great-$(d-1)$-spheres 
on the unit sphere $\SS^d$ (embedded in $\RR^{d+1}$).
The terms ``distance'' and ``$k$-level'' 
generalize in a natural way.

\section{Our Results}

In the first part of this paper 
we consider arrangements of great-circles on the sphere 
and investigate the average complexity of the $k$-level 
when the southpole is chosen uniformly at random among the cells.
This question was raised by Barba, Pilz, and Schnider while sharing a
pizza~\cite[Question~4.2]{BarbaPilzSchnider2019}.

In Section~\ref{sec:proof_k_zone} we
prove the following bound on the average complexity.

\begin{theorem}
\label{Thm:average_k-level}
Let $\CC$ be a simple arrangement of $n$ great-circles.
For $k < n/3$ the expected size of the $k$-level is 
at most $4e \cdot (k+2)^2$
when the southpole is chosen uniformly at random among the cells of $\CC$.
\end{theorem}

The condition $k < n/3$ is needed for Lemma~\ref{Lem:cyclic_intervals}
as for larger $k$ we would have to double the multiplicative constant. 
However, for $k$
in $\Omega(n^{3/5})$ the stated bound is implied by the $O(nk^{1/3})$
bound on the maximum size of a $k$-level. Still it is remarkable that the
bound is independent of the number $n$ of great-circles in the
arrangement.
\medskip

In the second part,
we investigate arrangements of randomly chosen great-circles.  
Here we propose the following model of randomness.  
On $\SS^2$ we have the duality between points 
and great-circles
(each antipodal pair of points
defines the normal vector of the plane containing a great-circle).
Since we can choose points uniformly at random from $\SS^2$,
we get random arrangements of great-circles.  
The duality generalizes to higher dimensions
so that we can talk about random arrangements on
$\SS^d$ for a fixed dimension $d \ge 2$. Using the
duality between antipodal pairs of points on $\SS^d$ and 
great-$(d-1)$-spheres, we prove the following bound on the
expected size of the $k$-level in this random model 
(the proof can be found in Section~\ref{sec:proof_great_circles}).
Again the bound does not depend on the size of the arrangement.

\begin{theorem}
\label{thm:random_circles}
Let $d \ge 2$ be fixed. In an arrangement of $n$ great-$(d-1)$-spheres chosen uniformly at
  random on the unit sphere $\SS^d$ (embedded in $\RR^{d+1}$),  
  the expected size of the $k$-level
  is of order $\Theta((k+1)^{d-1})$ for all $k \le n/2$. 
\end{theorem}

\section{Proof of Lemma~\ref{thm:ckn_bound_higher}}
\label{sec:proof_ckn_bound}

As hinted in Matou\v{s}ek's book \cite[Exercise~6.4.2]{Matousek2002_book}, 
we use the method of Clarkson and Shor \cite{ClarksonShor1989} 
to prove Lemma~\ref{thm:ckn_bound_higher}.

Let $\LL$ be an arrangement of $n$ lines in $\RR^2$ and let
$\ell \in \LL$ be a fixed line. For any $j=0,1,\ldots,n-1$ denote
by $\Zj$ the set of vertices of $\LL$ contained in the $(\le j)$-zone
$Z_{\le j}(\ell,\LL)$ of $\ell$ and lying strictly above $\ell$. 
In other words, $v\in\Zj$ if there is a simple path $P_v$ 
in the halfplane $\ell^+$ from
$v$ to $\ell$ whose interior has at most $j$ intersections with lines
from $\LL$.

Let $R$ be a random sample of lines from $\LL$ 
where $\ell \in R$ and each line $\ell' \neq \ell$
independently belongs to $R$ with probability $p:=\frac{1}{j+2}$.  The probability that a
vertex $v \in \Zj$ is present in the induced subarrangement $\LL(R)$ and appears at
distance~$0$ from $\ell$ is at least
$(\frac{1}{j+2})^2 \cdot (1-\frac{1}{j+2})^{r}$, where $0 \le r \le j$
denotes the distance of $v$ from $\ell$ in $\LL$.  
Figure~\ref{fig:ckn_bound} gives an illustration.
Note that
\[
  \left(1-\frac{1}{j+2}\right)^{r} \ge \left(1-\frac{1}{j+2}\right)^{j+1}
  = \left(\frac{j+1}{j+2}\right)^{j+1} = \left(1+\frac{1}{j+1}\right)^{-(j+1)} \ge 1/e.
\]

\begin{figure}[htb]
  \centering
    \includegraphics[page=1]{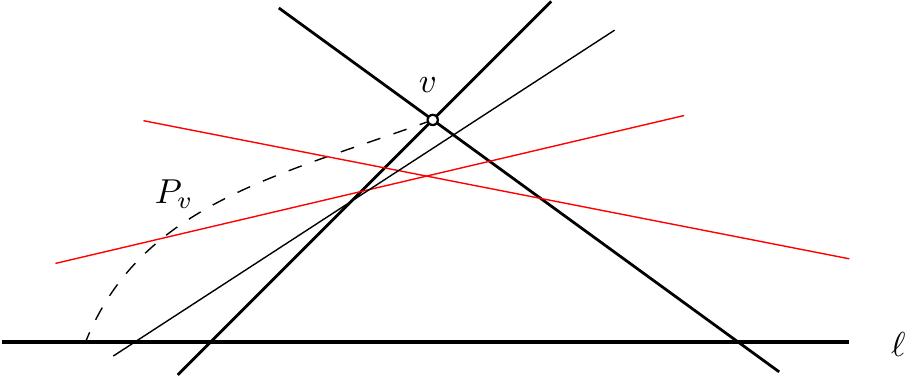}
  \caption{A path $P_v$ witnessing that $v$ belongs to the $(\leq j)$-zone of $\ell$ for all $j \ge 2$.}
  \label{fig:ckn_bound}
\end{figure}

Let $X$ be the number of vertices in the $0$-zone of $\ell$ in $\LL(R)$ that lie strictly above $\ell$.
For the
expectation of this random variable we have
\[
\EE(X) \ge \frac{1}{e} \left(\frac{1}{j+2}\right)^2 \cdot |\Zj|.
\]

An inductive argument, as used to show the classical zone theorem 
(see \cite[page~136]{ETH-Skript}),
shows there are at most $2n-3$ vertices
lying strictly above $\ell$ in the zone.
Hence, we have $X \leq 2 \cdot|R|$ and
\[
\EE(X) \leq 2 \cdot \EE(|R|) = 2 n p.
\]
The above inequalities imply
\[
|\Zj| \leq  e\cdot (j+2)^2 \cdot 2\cdot n \cdot p = 2 \cdot e\cdot (j+2) \cdot n.
\]
This concludes the proof of the theorem.

\section{Proof of Theorem~\ref{Thm:average_k-level}}
\label{sec:proof_k_zone}


For the proof of Theorem~\ref{Thm:average_k-level},
we fix a great-circle~$C$ from $\CC$ and 
denote the closures of the two hemispheres of $C$ on $\SS^2$ as $C^+$ and~$C^-$.
As an intermediate step, we
bound the size of the set $\FF_k(C^+)$ of pairs $(F,v)$, where $F$ is a cell
of $C^-$ touching $C$ and~$v$ is a vertex of $C^+$
whose distance to $F$ is $k$. 
We show $|\FF_k(C^+)| \le 2e \cdot (k+1)^2n$.
In the case $k=0$, 
vertex $v$ must be one of the $2n$ vertices on $C$
and $F$ is one of the two cells of $C^-$ which is adjacent to~$v$.
Hence, we obtain $|\FF_0(C^+)| \le 4n$.
It remains to deal with the general case $k \ge 1$.
Note that if $(v,F)\in \FF_k(C^+)$ then $v$ belongs to the $(\le k-1)$-zone of~$C$.

Consider a family~$\II$ of half-intervals in $\RR$, it consists of 
\emph{left-intervals} of the form $(-\infty,a]$ and
\emph{right-intervals} $[b,\infty)$.  A subset $J$ of $k$
half-intervals from~$\II$ is a $k$-\emph{clique} if there is a point
$p\in\RR$ that lies in all the half-intervals of $J$ but not in any
half-interval of~$\II\setminus J$.

\begin{lemma}
\label{Lem:intervals}
Any family $\HH$ of half-intervals in $\RR$ contains at most $k+1$ different $k$-cliques.
\end{lemma}

\begin{proof}
  For $p \in \RR$, let $l(p)$ be the number of left-intervals and
  $r(p)$ the number of right-intervals containing $p$. A point $p$ certifies a
  $k$-clique if and only if $l(p)+r(p)=k$. From the monotonicity of the functions
  $l$ and $r$ it follows that if $(l(p_1),r(p_1))=(l(p_2),r(p_2))$ for two
  points $p_1$ and $p_2$, then they are contained in the same intervals. Thus
  the number of $k$-cliques is at most the number of pairs $(l,r)$ such that
  $l+r=k$ and $l, r \ge 0$, which is $k+1$.
\end{proof}

The next lemma is a corresponding result for half-circles on the circle~$\SS^1$.

\begin{lemma}
\label{Lem:cyclic_intervals}
Any family $\HH$ of $n$ half-circles in $\SS^1$ with $n > 3k$ 
contains at most $k+1$ different $k$-cliques.
\end{lemma}

\begin{proof}
For this proof, we embed $\SS^1$ as the unit-circle in~$\RR^2$, 
which is centered at the origin~${\bf o}$.
We consider the set $X$ of all points from $\SS^1$,
which are contained in precisely $k$ of the half-circles of~$\II$,
and distinguish the following two cases.

Case 1: The origin~${\bf o}$ is not contained in the convex hull
of~$X$.  There is a line separating ${\bf o}$ from $X$ and 
rotational symmetry allows us to assume  that $X$ is contained in
$\Pi^+ = \{(x,y) \in \RR^2 \colon y>0\}$.  For each half-circle
$C \in \HH$, the central projection of $C \cap \Pi^+$ to the line
$y=1$ is a half-interval. Since $k$-cliques of $\HH$ and $k$-cliques of the
half-intervals are in bijection we get from Lemma~\ref{Lem:intervals}
that $\HH$ has at most $k+1$ different $k$-cliques.

Case 2: The origin~${\bf o}$ is contained in the convex hull of~$X$.
By Carath\'{e}odory's theorem,
we can find three points $p_1,p_2,p_3$
such that ${\bf o}$ lies in the convex hull of $p_1,p_2,p_3$.
Since each of the $n$ half-circles from $\HH$ contains at least one of these three points,
and each of these three points lies on precisely $k$ half-circles,
we have $n \le 3k$ -- a contradiction to  $n > 3k$.
\end{proof}

For a fixed vertex $v$ in the $(\leq k-1)$-zone of $C$ with
$v \in C^+$, let $\BB_{C^+}(v)$ be the set of cells~$F$ such that
$(F,v) \in\FF_k(C^+)$, in particular $\dist(F,v)=k$.

\medskip

\noindent{\bf Claim.} For $k \ge 1$, we have
 $|\BB_{C^+}(v)|\leq k$.

\begin{proof}
  Consider a great-circle $D \neq C$ from $\CC$.
  For a point $x \in C$, 
  we say that $(v,x)$ is \emph{$D$-separated} 
  if every path from $v$ to $x$ in $C^+$ intersects~$D$.
  The set of all \emph{$D$-separated} points forms 
  a half-circle $H_D$ on $C$.
  Let $\HH$ be the set of these half-circles,
  i.e., $\HH = \{ H_D : D \in \CC, D \neq C \}$.
  See Figure~\ref{Fig:Intervals}.
  
\begin{figure}[htb]
\centering
\includegraphics[scale=0.8,page=2]{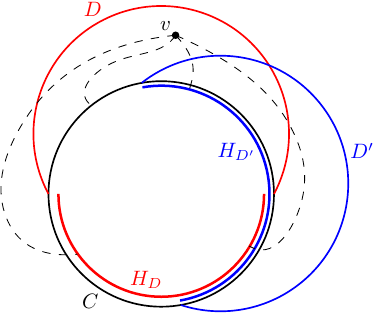}
\caption{An illustration of the cyclic half-circles $\HH$.}
\label{Fig:Intervals}
\end{figure}

  We claim that
  there is a bijection between $\BB_{C^+}(v)$ and the $(k-1)$-cliques
  in $\HH$.  Indeed, if the intersection of the half-circles of a
  clique~$K$, viewed as a subset of $C$, is $I_K$, then $I_K$ is the interval
  of $C$ which is reachable from $v$ by crossing the circles corresponding to
  the half-circles of $K$. If $F$ is a cell from $C^-$ at distance $k$
  from $v$, then $C$ and a subset of $k-1$ additional circles have to be
  crossed to reach $v$ from $F$, i.e., there is a $(k-1)$-clique in
  $\HH$ whose intersection is $F\cap C$. The number of
  $(k-1)$-cliques in $\HH$ is at most $k$ by Lemma~\ref{Lem:cyclic_intervals}.
\end{proof}

\medskip
\noindent{\bf Claim.} For $k \ge 1$, we have
$|\FF_k(C^+)| \le 2e \cdot k(k+1)n$.

\begin{proof}
  By definition, the set $\FF_k(C^+)$ is the set
  of pairs $(F,v)$ such that 
  $v \in C^+$ is in the $(\le k)$-zone of~$C$
  and $F \in \BB_{C^+}(v)$. 
  As already noted in Section~\ref{ssec:arrangements_greatcircles},
  the $(\le k)$-zone contains at most $4e \cdot (k+1)n$ vertices of $\CC$
  and at most $2e \cdot (k+1)n$ vertices in~$C^+$.
   From the above claim we have
  $|\BB_{C^+}(v)|\leq k$, hence we conclude that
  $|\FF_k(C^+)| \le 2e \cdot k(k+1)n$.
\end{proof}

To include the case $k=0$ we relax the bound to
$|\FF_k(C^+)| \le 2e \cdot (k+1)^2n$.
Since $C$ was chosen arbitrarily among all great-circles from $\CC$
and $C^+$ was chosen arbitrarily among the two hemispheres of~$C$,
the upper bound from the above claim holds for any induced hemisphere of $\CC$.
For the union $\FF_k$ of the $\FF_k(C^+)$ over all 
the $2n$ choices of the hemisphere $C^+$,
we have
\[
    |\FF_k| \le \sum_{\text{$C^+$ hemisphere}} |\FF_k(C^+)| \le 4e  (k+1)^2 n^2. 
\]

\begin{proof}[Proof of Theorem \ref{Thm:average_k-level}]
  The $k$-level with the southpole chosen
  in cell~$F$ consists of the vertices at
  distance $k$ from~$F$.  Thus, the expected complexity of the
  $k$-level when choosing $F$ uniformly at random equals 
  $|\FF_k|$ divided by the number of cells.  Since the number of
  cells in an arrangement of $n$ great-circles is $2\binom{n}{2}+2$
  and $|\FF_k| \le 4e  (k+1)^2  n^2$,
  we can conclude the statement from 
  \[
    \frac{4e\cdot (k+1)^2 \cdot n^2}{2\binom{n}{2}+2} 
    \le 4e \cdot (k+1)^2 \cdot \frac{n}{n-1}
    \le 4e \cdot (k+2)^2 \cdot \underbrace{\frac{k+1}{k+2} \cdot \frac{n}{n-1}}_{\le 1}.  
     \qedhere 
  \]
\end{proof}

\section{Proof of Theorem~\ref{thm:random_circles}}
\label{sec:proof_great_circles}
\def\vol{\text{\rm Vol}}

Let $\CC$ be a simple arrangement of $n$ great-$(d-1)$-spheres on the unit sphere
$\SS^d= \{ x \in \RR^{d+1} : \|x\| = 1\}$ with center ${\bf o}=(0,\ldots,0)$ in
$\RR^{d+1}$. For a vertex $v$ of the arrangement, let $\phi_\CC(v)$
denote the number of great-$(d-1)$-spheres of $\CC$ that are crossed by the geodesic
arc from $v$ to the south-pole ${\bf s}=(0,\ldots,0,-1)$ of the
sphere. The set of vertices $v$ of $\CC$ with $\phi_\CC(v)=k$ is
denoted $V_k(\CC)$.

When $\CC$ is projected to a $d$-dimensional plane $H$ with the origin
${\bf o}$ as center of projection, we obtain an arrangement $\AA$
of hyperplanes in $\RR^d$.  Moreover, if the south pole $\bf s$ is projected
to a point ``at infinity'' of $H$, say to $(0,\ldots,0,-\infty)$, then, for every point
$p$ in $\SS^d$, the circle in $\SS^d$ containing the geodesic arc from $p$ to $\bf{s}$ is projected to the
``vertical'' line through $p$, i.e., the line $p+(0,\ldots,0,\lambda)$.
The geodesic is projected to one of the two rays starting from $p$ on this line. In particular, all
vertices $v$ of $\CC$ with $\phi_\CC(v)=k$ are projected to vertices of $\AA$
either at level $k$ or $n-k-d$.

Let $\CC$ be an arrangement of randomly chosen great-$(d-1)$-spheres and let 
$\BB$ be a subset of size $d$ in $\CC$. Note that with probability $1$, the random
great-sphere-arrangement is in general position, and simple, i.e., no more than $d$ great-spheres
intersect in a common point. Choose $p'$ as one of the
two intersection points of the great-$(d-1)$-spheres in $\BB$.  Now consider the
arrangement $\CC' = \CC - \BB$ and note that $(\CC',p')$ can be viewed as a
random arrangement of great-$(d-1)$-spheres together with a random point on
$\SS^d$.  Hence, to estimate the expected size of $V_k(\CC)$, we can 
estimate the probability that $\phi_{\CC'}(p')=k$. This is the purpose of the
following lemma.
\begin{lemma}
\label{lemma:random_circles_intersect_arc}
Let $\CC$ be an arrangement of $n$ great-$(d-1)$-spheres chosen uniformly at random on
the unit sphere $\SS^d$ (embedded in $\RR^{d+1}$ and centered at the origin).  
Let $p$ be an additional point chosen uniformly at random from $\SS^d$, and
let $A$ be the geodesic arc from $p$ to the south pole 
on $\SS^d$.
For all $k \le n/2$, the probability~$q_k$ that exactly $k$ great-$(d-1)$-spheres
from~$\CC$ intersect $A$ is in $\Theta((k+1)^{d-1}/n^d)$. More precisely, it satisfies
\[
\frac{
2^{d-1} \rho \pi(k+1)\risingfactorial{d-1}(n-k+1)\risingfactorial{d-1}}{(n+1)\risingfactorial{2d-1}}
\le q_k \le 
\min
\left \{
\frac{\rho\pi}{n+1},
\frac{\rho\pi^d(k+1)\risingfactorial{d-1}}{(n+1)\risingfactorial{d}}
\right \},
\]
where $a\risingfactorial{b}=a(a+1)\cdots(a+b-1)$ denotes the rising factorial
and 
$
\rho = \rho_d = \frac{\area_{d-1}(\SS^{d-1}) }{\area_{d}(\SS^{d})} 
= \frac{ \Gamma(\frac{d+1}{2})}{\pi^{1/2} \Gamma(\frac{d}{2})}
$ 
only depends on the dimension~$d$.
\end{lemma}

\begin{proof}   
Denote by $\phi$ the length of the geodesic arc $A$ on $\SS^d$
from~$p$ to $\bf s$, i.e., $\phi$ is the angle between the two rays 
emanating from $\bf o$ towards $\bf s$ and $p$.
Note that -- independent from the dimension~$d$ -- 
the three points $\bf o$, $\bf s$, and $p$ lie in a 2-dimensional plane
which also contains the geodesic arc~$A$.

Point $p$ lies on a $(d-1)$-sphere $C$  
of radius $\sin(\phi)$ in the $d$-dimensional hyperplane defined by the equation $x_d=-\cos(\phi)$.
Figure~\ref{fig:sphere} gives an illustration for the case $d=2$,
where $C$ is a circle.

\begin{figure}[htb]
\centering
\includegraphics[scale=0.8]{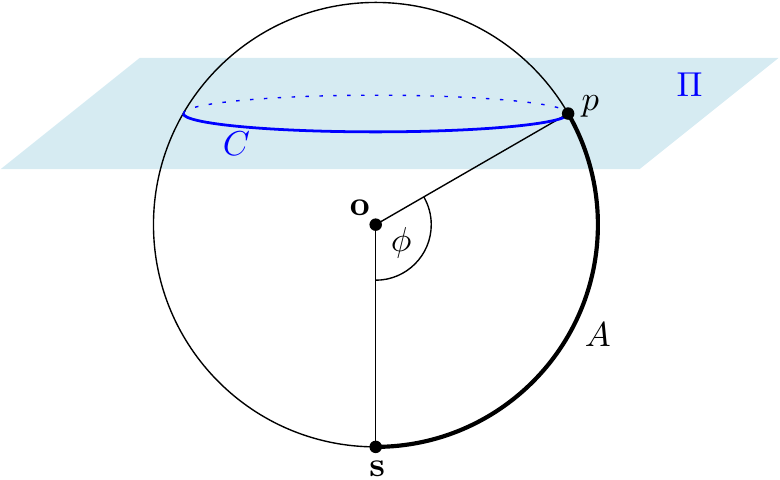}
\caption{Illustrating the definitions of $A$, $C$, and $\Pi$ depending on $p$.} 
\label{fig:sphere}
\end{figure}

The probability that the arc $A$ defined by the random point $p$ is intersected by exactly
$k$ great-$(d-1)$-spheres from the random arrangement~$\CC$ is 
\begin{align*}
q_k = 
\int_{\phi = 0}^{\pi} 
\underbrace{
\frac{\vol_{d-1}(\SS^{d-1}) \sin^{d-1}(\phi)}{\vol_{d}(\SS^{d})} }_\text{density at angle $\phi$}
\cdot
\underbrace{\binom{n}{k}(\phi / \pi)^k (1-\phi/\pi)^{n-k}}_{\text{chosen great-}(d-1)\text{-spheres intersect $A$}}
d\phi.
\end{align*}

\noindent
This can be rewritten as
\[
  q_k = \rho\cdot \binom{n}{k} \cdot \int_{\phi = 0}^{\pi} \sin^{d-1}(\phi) \cdot
  (\phi / \pi)^k (1-\phi/\pi)^{n-k} d\phi,
\]
where
$\rho = \rho(d) = \frac{\vol_{d-1}(\SS^{d-1}) }{\vol_{d}(\SS^{d})} = \frac{
  \Gamma(\frac{d+1}{2})}{\pi^{1/2} \Gamma(\frac{d}{2})}$ is a constant only depending on~$d$. 
The latter equation
follows from $ \vol_{d}(\SS^d) = 2\pi^{\frac{d+1}{2}}/\Gamma(\frac{d+1}{2})$, where
$\Gamma(x)$ is the Euler gamma function (see e.g.\ \cite{wiki:N-sphere}).

\goodbreak
\medskip

In the following we give upper and lower bounds for $q_k$.
The Euler beta function $B$ 
turns out to be the tool to evaluate the integrals:
\[
B(a+1,b+1) =
\int_{t = 0}^{1} t^{a} (1-t)^{b} dt =
\frac{a!\cdot b!}{(a+b+1)!}.
\]
For this identity and more information see for example~\cite{wiki:Beta_function}.
\medskip

To show the first upper bound on $q_k$, we bound the integral above as follows:
Since  $\sin(\phi) \le 1$ holds for every $\phi \in [0,\pi]$,
we have 
\begin{align*}
q_k &\le \rho \binom{n}{k}
\int_{\phi = 0}^{\pi} 
(\phi / \pi)^k (1-\phi/\pi)^{n-k}
d\phi
= 
\rho \pi 
\binom{n}{k}
\int_{t = 0}^{1} 
t^{k} (1-t)^{n-k}
dt\\
&=
\rho \pi 
\binom{n}{k}B(k+1,n-k+1)
=
\rho \pi 
\cdot
\frac{n!}{k!(n-k)!}
\cdot
\frac{k!(n-k)!}{(n+1)!}
=
\rho \pi
\cdot
\frac{1}{n+1}.
\end{align*}

\medskip
\noindent
Towards the second upper bound on $q_k$,
we use the fact that $\sin(\phi) \le \phi$ holds for every $\phi \in [0,\pi]$:
\begin{align*}
q_k 
&\le 
\rho \pi^{d-1}
\binom{n}{k}
\int_{\phi = 0}^{\pi} 
(\phi / \pi)^{k+d-1} (1-\phi/\pi)^{n-k}
d\phi\\
&=
\rho \pi^d
\binom{n}{k}
\int_{t = 0}^{1} 
t^{k+d-1} (1-t)^{n-k}
dt\\
&=
\rho \pi^d
\cdot
\frac{n!}{k!(n-k)!}
\cdot
\frac{(k+d-1)!(n-k)!}{(n+d)!}
=
\rho \pi^d
\cdot
\frac{(k+1)\risingfactorial{d-1}}{(n+1)\risingfactorial{d}}.
\end{align*}

\medskip
\noindent
To show the lower bound on $q_k$, we split the integral in two parts:
Since  $\sin(\phi) \ge 2 \cdot \frac{\phi}{\pi}$ 
holds for every $\phi \in [0,\pi/2]$ and
$\sin(\phi) \ge 2\cdot (1-\frac{\phi}{\pi})$ holds for every $\phi \in [\pi/2,\pi]$,
we have 
\begin{align*}
q_k 
&\ge 
2^{d-1} \rho
 \binom{n}{k} 
\left[
\int_{\phi = 0}^{\pi/2} 
(\phi / \pi)^{k+d-1} (1-\phi/\pi)^{n-k}
d\phi
 +
\int_{\phi = \pi/2}^{\pi} 
(\phi / \pi)^{k} (1-\phi/\pi)^{n-k+d-1}
d\phi
\right]
\\
& \ge
2^{d-1} \rho
 \binom{n}{k} 
\int_{\phi = 0}^{\pi} 
(\phi / \pi)^{k+d-1} (1-\phi/\pi)^{n-k+d-1}
d\phi\\
& =
2^{d-1} \rho
{\pi}
\binom{n}{k}
\int_{t = 0}^{1} 
t^{k+d-1} (1-t)^{n-k+d-1}
dt
\\
&=
2^{d-1} \rho
{\pi}
\cdot
\frac{n!}{k!(n-k)!}
\cdot
\frac{(k+d-1)!(n-k+d-1)!}{(n+2d-1)!}
\\
&=
\frac{
2^{d-1} \rho\pi(k+1)\risingfactorial{d-1}(n-k+1)\risingfactorial{d-1}}{(n+1)\risingfactorial{2d-1}}.
\end{align*}
This completes the proof of Lemma~\ref{lemma:random_circles_intersect_arc}.
\end{proof}

\begin{proof}[Proof of Theorem~\ref{thm:random_circles}]
Consider an arrangement $\CC$ of $n+d$ great-$(d-1)$-spheres $C_1,\ldots,C_{n+d}$ 
chosen uniformly and independently at random from $\SS^d$.
Let $p$ be a vertex of $\CC$ chosen uniformly at random
(i.e., one of the two points of intersection of $d$ great-$(d-1)$-spheres
$C_{i_1},\ldots,C_{i_d}$ chosen u.a.r.\ from $\CC$). 
Note that $p$ is a u.a.r.\ chosen point from $\SS^d$.

We now apply Lemma~\ref{lemma:random_circles_intersect_arc}
with $p$ and $\CC_p := \CC - \{C_{i_1},\ldots,C_{i_d}\}$. 
Point $p$ is separated from
$\bf s$ by $k$ great-$(d-1)$-spheres from $\CC_p$ with probability $q_k = \Theta(k^{d-1}/n^d)$. 
Since $p$ is chosen uniformly at random among the $2\binom{n+d}{d}$ vertices of $\CC$, 
we obtain the desired bound of $\Theta(k^{d-1})$ for the number of
vertices at distance $\leq k$ from $p$. Rotating the sphere such
that $p$ becomes the southpole does not affect the fact that
the arrangement of circles not incident to $p$ is a sample
from the uniform dirtribution.

\end{proof}

\section{Discussion}

Theorem~\ref{Thm:average_k-level} is about arrangements of
great-circles. All the elements of the proof, however,
carry over to great-pseudocircles whence the result could also be
stated for arrangements of great-pseudocircles. Projective
arrangements of lines are obtained by antipodal identification
from arrangements of great-circles. Hence, if you pick a cell u.a.r.\ in a
projective arrangement of lines (pseudo-lines) the the expected
number of vertices at distance $k$ from the cell is as in 
Theorem~\ref{Thm:average_k-level}. If the projection $\Psi_\Pi$ is
used to project an arrangements $\CC$ of great-pseudocircles to an
Euclidean arrangement $\LL$ on $\Pi$ such that the south-poles 
coincide, then the $k$-level of $\CC$ corresponds to the union of the
$k$- and the $(n-k-2)$-level of $\LL$.

\medskip
With respect to lower bounds we would like to know
the answer to:
\begin{question}
Is there a family of  arrangements  where 
the expected size of the middle level is superlinear 
when the southpole is chosen uniformly at random?
\end{question}

Recursive constructions from \cite{EdelsbrunnerWelzl1985} and
\cite{ErdosLSS1973} show that the size of the $(n/2-s)$-level can be in
$\Omega(n \log n)$ for any fixed $s$.  Nevertheless computer
experiments suggest that if we choose a random southpole for these
examples the expected size of the middle level drops to be linear.
\medskip

Theorem~\ref{thm:random_circles} deals with the average size of the
$k$-level in arrangements of randomly chosen great-circles. In our
model, great-circles are chosen independently and uniformly at random
from the sphere.  Since point sets, line arrangements, and
great-circle arrangements are in strong correspondence the bound from
Theorem~\ref{thm:random_circles} also applies to $k$-sets in point
sets and $k$-levels of line arrangements from a specific random
distribution.

In the context of Erd\H{o}s--Szekeres-type problems, several articles
made use of point sets which are sampled uniformly at random from a
convex shape
\cite{BaranyFueredi1987,Valtr1995,BaloghGAS2013,BalkoScheucherValtr2019}.
Also the average size of the convex hull ($0$-level) is well-studied
for sets of points which are sampled uniformly at random from a convex
shape~$K$.  If $K$ is a disk, the convex hull has expected size
$O(n^{1/3})$, and if $K$ is a convex polygon with $k$ sides, the
expected size is $O(k\log n)$
\cite{HarPeled2011,PreparataShamos1985,Raynaud1970,RenyiSulanke1963}.
In particular, the expected size of the convex hull is not constant,
which is a substantial contrast to our setting.  In fact, our setting
appears to be closer to the setting of random order types, for which
the expected size of the convex hull was recently shown to be $4+o(1)$
\cite{GoaocWelzl2020}.  Hence it would be very interesting to obtain
bounds on the average number of $k$-sets also in this setting.  Last
but not least, Edelman \cite{Edelman1992} showed that the expected
number of $k$-sets of an allowable sequence is of order
$\Theta(\sqrt{kn})$.

\bibliography{references}
\bibliographystyle{alphaabbrv-url}

\end{document}